\documentclass[preprint]{elsarticle}
\usepackage{amsthm}
\usepackage{ifthen}
\usepackage{color}		
\usepackage{epsfig}
\usepackage{amsmath}
\usepackage{amssymb}
\usepackage[rightcaption]{sidecap}

\usepackage[show]{todo}

\newtheorem{theorem}{Theorem}

\newtheorem{lemma}[theorem]{Lemma}

\label{in:defs} 
\newcommand{\eg}{\textit{e.g.,}}

\newcommand{\ie}{\textit{i.e.,}}
\newcommand{\ai}{\textit{et al.}}

\newcommand{\Rp}{R}

\newcommand{\no}[1]{}

\newcommand{\algprm}[1]{\textsc{#1}} 

\newcommand{\rank}{\algprm{Rank}}
\newcommand{\select}{\algprm{Select}}
\newcommand{\raiz}{\algprm{Root}}
\newcommand{\wsum}{\algprm{Wsum}}
\newcommand{\wmin}{\algprm{Wmin}}
\newcommand{\wsucc}{\algprm{Wsucc}}

\newcommand{\frank}[3]{\ensuremath{\rank(#1, #2, #3)}}
\newcommand{\fselect}[3]{\ensuremath{\select(#1, #2, #3)}}

\newcommand{\imp}{\ensuremath{\textsc{imp}(y_0, y_1)}}
\newcommand{\impN}{\ensuremath{\textsc{imp}}}
\newcommand{\conta}{\ensuremath{\algprm{Count}(Q)}}
\newcommand{\contaN}{\ensuremath{\algprm{Count}}}
\newcommand{\soma}{\ensuremath{\algprm{Sum}(Q)}}
\newcommand{\somap}{\ensuremath{\algprm{Sum}'(Q)}}

\newcommand{\somapN}{\ensuremath{\algprm{Sum}'}}

\newcommand{\avga}{\ensuremath{\algprm{Avg}(Q)}}
\newcommand{\vara}{\ensuremath{\algprm{Var}(Q)}}
\newcommand{\somaN}{\ensuremath{\algprm{Sum}}}
\newcommand{\avgaN}{\ensuremath{\algprm{Avg}}}
\newcommand{\varaN}{\ensuremath{\algprm{Var}}}
\newcommand{\mina}{\ensuremath{\algprm{Min}(Q)}}
\newcommand{\maxa}{\ensuremath{\algprm{Max}(Q)}}
\newcommand{\minaN}{\ensuremath{\algprm{Min}}}
\newcommand{\maxaN}{\ensuremath{\algprm{Max}}}
\newcommand{\preda}{\ensuremath{\algprm{Pred}(w, Q)}}
\newcommand{\succa}{\ensuremath{\algprm{Succ}(w, Q)}}
\newcommand{\predaN}{\ensuremath{\algprm{Pred}}}
\newcommand{\succaN}{\ensuremath{\algprm{Succ}}}
\newcommand{\maja}{\ensuremath{\algprm{Majority}(\alpha,Q)}}
\newcommand{\majaN}{\ensuremath{\algprm{Majority}(\alpha)}}
\newcommand{\quantile}{\ensuremath{\algprm{Quantile}(k,Q)}}
\newcommand{\quantileN}{\ensuremath{\algprm{Quantile}}}
\newcommand{\Pe}{\ensuremath{\mathbf{P}}}

\begin{document}
\title{Space-Efficient Data-Analysis Queries on Grids}
\author[chile]{Gonzalo Navarro\fnref{fn1}}
\ead{gnavarro@dcc.uchile.cl} 
\address[chile]{Dept. of Computer Science, University of Chile.}\fntext[fn1]{Author
supported in part by Millennium Institute for Cell Dynamics and Biotechnology (ICDB), Grant ICM P05-001-F, Mideplan, Chile.}
\author[chile]{Yakov Nekrich\fnref{fn1}}
\ead{yakov.nekrich@googlemail.com}
\author[ist,inesc]{Lu\'{i}s M.\ S.\ Russo\fnref{fn2}}
\ead{luis.russo@ist.utl.pt} 
\address[ist]{Instituto Superior T\'{e}cnico
  - Universidade T\'{e}cnica de Lisboa
  (IST/UTL).}
\address[inesc]{INESC-ID / KDBIO}
\fntext[fn2]{Author
supported by FCT through projects TAGS PTDC/EIA-EIA/112283/2009, HELIX
PTDC/EEA-ELC/113999/2009 and the PIDDAC Program funds (INESC-ID
multiannual funding).}

%
%
%
%
\begin{abstract}
\label{sec:abstract}
  We consider various data-analysis queries on two-dimensional points.
  We give new space/time tradeoffs over previous work on geometric
  queries such as dominance and rectangle visibility, and on semigroup
  and group queries such as sum, average, variance, minimum and
  maximum. We also introduce new solutions to queries less frequently
  considered in the literature such as two-dimensional quantiles,
  majorities, successor/predecessor, mode, and various top-$k$ queries, 
  considering static and dynamic scenarios.
\end{abstract}
\maketitle
\section{Introduction}

Multidimensional grids arise as natural representations to support conjunctive 
queries in databases \cite{BCKO08}. Typical queries such as ``find all the 
employees with age between $x_0$ and $x_1$ and salary between $y_0$ and $y_1$'' 
translate into a two-dimensional range reporting query on coordinates age and 
salary. More generally, such a grid representation of the data is useful to
carry out a number of {\em data analysis} queries over large repositories.
Both space and time efficiency are important when analyzing the performance 
of data structures on massive data. However, in cases of very large data 
volumes the space usage can be even more important than the time needed to 
answer a query. More or less space usage can make the difference between 
maintaining all the data in main memory or having to resort to disk, which is
orders of magnitude slower.

In this paper we study various problems on two-dimensional grids that
are relevant for data analysis, focusing on achieving good time performance
(usually polylogarithmic) within the least possible space (even succinct for 
some problems).

Counting the points in a two-dimensional range $Q=[x_0,x_1] \times
[y_0,y_1]$, \ie\ computing \conta, is arguably the most primitive
operation in data analysis. Given $n$ points on an $n\times n$ grid, one can 
compute \contaN\ in time $O(\log n/\log\log n)$ using ``linear'' space, $O(n)$
{\em integers} \cite{Nek09}. This time is optimal within space
$O(n\,\mathrm{polylog}(n))$ \cite{Pat07}, and it has been matched using
asymptotically minimum (\ie\ succinct) space, $n+o(n)$ integers, by Bose 
\ai~\cite{BHMM09}. 

The $k$ points in the range can be reported in time $O(\log\log n +k)$ using 
$O(n\log^\epsilon n)$ integers, for any constant $\epsilon>0$ \cite{ABR00}. 
This time is optimal within space $O(n\,\mathrm{polylog}(n))$, by reduction
from the colored predecessor problem \cite{PT06}. With $O(n)$ integers, the 
time becomes $O((k+1)\log^\epsilon n)$ \cite{CLP11}. 
Within $n+o(n)$ integers space, one can achieve time
$O(k\frac{\log n}{\log \log n})$ \cite{BHMM09}.

We start with two geometric problems, {\em
  dominance} and {\em rectangular visibility}. These
enable data analysis queries such as ``find the employees with worst
productivity-salary combination within productivity range $[x_0,x_1]$
and salary range $[y_0,y_1]$'', that is, such that no less productive
employee earns more.

The best current result for the 4-sided variant of these problems
(\ie\ where the points are limited by a general rectangle $Q$) is a
dynamic structure by Brodal and Tsakalidis \cite{TB11}. It requires
$O(n\log n)$-integers space and reports the $d$ dominant/visible
points in time $O(\log^2 n + d)$. Updates take $O(\log^2 n)$
time. They achieve better complexities for simpler variants of the
problem, such as some 3-sided variants. 

Our results build on the {\em wavelet tree} \cite{GGV03}, a succinct-space 
variant of a classical structure by Chazelle \cite{Cha88}. The wavelet tree
has been used to handle various geometric problems, \eg\
\cite{MN07,GPT09,BHMM09,BLNS10,BCN10,GNP11}.  We show in
Section~\ref{sec:geometric-queries} how to use wavelet trees to solve dominance
and visibility problems using $n+o(n)$ integers space and $O((d+1)\log n)$
time. 
The dynamic version also uses succinct space and requires
$O((d+1)\log^2 n / \log\log n)$ time, carrying out updates in time
$O(\log^2 n / \log\log n)$. Compared to the best current result \cite{TB11},
our structure requires succinct space instead of $O(n\log n)$ integers,
offers better update time, and has a comparable query time (being faster for 
small $d = O(\log\log n)$).

The paper then considers a wide range of queries we call
``statistical'': The points have an associated value in $[0,W) =
[0,W-1]$ and, given a rectangle $Q$, we consider the following
queries:
\begin{description} 
\item[\rm $\somaN$/$\avgaN$/$\varaN$:] The sum/average/variance of the values 
in $Q$ (Section~\ref{sec:sum}).
\item[\rm $\minaN$/$\maxaN$:] The minimum/maximum value in $Q$
(Section~\ref{sec:rmq}).
\item[\rm $\quantileN$:] The $k$-th smallest value in $Q$
(Section~\ref{sec:quantile}).
\item[\rm $\majaN$:] The values appearing with relative frequency $>\alpha$ in 
$Q$ (Sections~\ref{sec:majf} and \ref{sec:majv}).
\item[\rm $\succaN$/$\predaN$:] The successor/predecessor of a value $w$ in $Q$
(Section~\ref{sec:suc}).
\end{description}

These operations enable data-analysis queries such as ``the average 
salary of employees whose annual production is between $x_0$ and $x_1$ and 
whose age is between $y_0$ and $y_1$''. The minimum operation can be used to 
determine ``the employee with the lowest salary'', in the previous 
conditions. The $\alpha$-majority operation can be used to compute ``which are 
frequent ($\ge 20\%$) salaries''. Quantile queries enable us  to determine 
``which are the 10\% highest salaries''. Successor queries can be used 
to find the smallest salary over $\$100{,}000$ among those employees.

Other applications for such queries are frequently found in Geographic 
Information Systems (GIS), where the points have a geometric interpretation 
and the values can be city sizes, industrial production, topographic heights, 
and so on. Yet another application comes from Bioinformatics, where 
two-dimensional points with intensities are obtained from DNA microarrays, and 
various kinds of data-analysis activities are carried out on them.
See Rahul \ai~\cite{RGJR11} for an ample discussion on some of these 
applications and several others.

A popular set of statistical queries includes range sums, averages, and
maxima/minima. Willard \cite{Wil85} solved two-dimensional range-sum queries 
on finite groups within $O(n\log n)$-integers space and $O(\log n)$ time. 
This includes $\somaN$ and is easily extended to $\avgaN$ and $\varaN$. Alstrup
\ai~\cite{ABR00} obtained the same complexities for the semigroup model, which 
includes $\minaN$/$\maxaN$. The latter can also be solved in constant time 
using $O(n^2)$-integers space \cite{AFL07,BDR11}. Chazelle \cite{Cha88} showed
how to reduce the space to $O(n/\epsilon)$ integers and achieve time
$O(\log^{2+\epsilon} n)$ on semigroups, and $O(\log^{1+\epsilon} n)$ for
the particular case of $\minaN/\maxaN$.

On this set of queries our contribution is to achieve good time complexities
within linear and even succinct space. This is relevant to 
handle large datasets in main
memory. While the times we achieve are not competitive when using $O(n)$ 
integers or more, we manage to achieve polylogarithmic times within just
$n\log n + o(n\log n)$ bits on top of the bare coordinates and values,
which we show is close to the information-theoretic minimum space necessary
to represent $n$ points.

As explained, we use wavelet
trees. These store bit vectors at the nodes of a tree that decomposes the 
$y$-space of the grid. The vectors track down the points, sorted on top by 
$x$-coordinate and on the bottom by $y$-coordinate. We enrich wavelet trees 
with extra data aligned to the bit vectors, which speeds up the computation of 
the statistical queries. Space is then reduced by sparsifying these extra data.

We also focus on more sophisticated queries, for which fewer results exist,
such as quantile, majority, and predecessor/successor queries.

In one dimension, the best result we know of for quantiles queries is a 
linear-space structure by Brodal and J{\o}rgensen \cite{BJ09}, which finds the 
$k$-th element of any range in an array of length $n$ in 
time $O(\log n / \log\log n)$, which is optimal.

An $\alpha$-majority of a range $Q$ is a value that occurs more than
$\alpha \cdot \conta$ times inside $Q$, for some $\alpha \in [0,1]$.
The $\alpha$-majority problem was previously considered in one and two
dimensions~\cite{KN08,DHMNS11,springerlink:10.1007/978-3-642-24583-1_29}.  
Durocher \ai~\cite{DHMNS11} solve
one-dimensional $\alpha$-range majority queries in time $O(1/\alpha)$
using $O(n(1+\log(1/\alpha)))$ integers.  Here $\alpha$ must be chosen
when creating the data structure. A more recent result, given by
Gagie \ai~\cite{springerlink:10.1007/978-3-642-24583-1_29}, 
obtains a structure of $O(n^2(H + 1) \log (1/\alpha))$ bits for a
dense $n \times n$ matrix (\ie\ every position contains an
element), where $H$ is the entropy of the distribution of elements.
In this case $\alpha$ is also chosen at indexing time, and
the structure can answer queries for any $\beta \geq \alpha$. The resulting
elements are not guaranteed to be $\beta$-majorities, as the list may
contain false positives, but there are no false negatives.

Other related queries have been studied in two dimensions.
Rahul \ai~\cite{RGJR11} considered a variant of \quantileN\ 
where one reports the {\em top-$k$ smallest/largest values} in a range. 
They obtain $O(n\log^2 n)$-integers space and $O(\log n + k\log\log n)$ time. 
Navarro and Nekrich \cite{NN12} reduced the space to $O(n/\epsilon)$ integers,
with time $O(\log^{1+\epsilon} n + k\log^\epsilon n)$.
Durocher and Morrison \cite{DM11} consider the {\em mode} (most repeated value)
in a two-dimensional range. Their times are sublinear but 
super-polylogarithmic by far.

Our contribution in this case is a data structure of $O(n\log n)$ integers able
to solve the three basic queries in time $O(\log^2 n)$. The space can be
stretched up to linear, but at this point the times grow to the form
$O(n^\epsilon)$. Our solution for range majorities lets $\alpha$ to be
specified at query time. For the case of $\alpha$ known at indexing time,
we introduce a new linear-space data structure that answers queries in
time $O(\log^3 n)$, and up to $O(\log^2 n)$ when using $O(n\log n)$ integers 
space.

\begin{table}[bt]
\begin{center}
\scalebox{0.65}{
\begin{tabular}{l@{~}|@{~}c@{~}|@{~}c@{~}|@{~}c@{~}|@{~}l}
Operation & Space per point (bits) & Time & Time in linear space & Source \\
\hline
$\somaN$, $\avgaN$, $\varaN$ 
	& $\log n(1+1/t)$ 
		& $O(\min(t \log W,\log n)t \log W \log n)$ 
			& $O(\log^3 n)$ 
				& Thm.~\ref{thm:soma} \\
$\minaN$, $\maxaN$ 
	& $\log n (1+1/t)$ 
		& $O(\min(t \log m,\log n)t\log n)$ 
			& $O(\log^2 n)$ 
				& Thm.~\ref{thm:mina} \\
$\majaN$, fixed
	& $\log n(2+1/t) +\log m$
		& $O(t\log m \log^2 n)$
			& $O(\log^3 n)$
				& Thm.~\ref{thm:fmaja} \\
$\quantileN$ 
	& $\log n \log_\ell m + O(\log m)$ 
		& $O(\ell \log n \log_\ell m)$ 
			& $O(n^\epsilon)$
				& Thm.~\ref{thm:quantile} \\
$\majaN$, variable
	& $\log n \log_\ell m + O(\log m)$ 
		& $O(\frac{1}{\alpha}\ell \log n \log_\ell m)$ 
			& $O(n^\epsilon)$
				& Thm.~\ref{thm:maja} \\
$\succaN$, $\predaN$
	& $\log n \log_\ell m + O(\log m)$
		& $O(\ell \log n \log_\ell m)$
			& $O(n^\epsilon)$
				& Thm.~\ref{thm:preda} \\
\end{tabular}
}
\end{center}
\caption{Our static results on statistical queries, for $n$ two-dimensional 
points with associated values in $[0,W)$; $m = \min(n,W)$; $2 \le \ell \le u$
and $t \ge 1$ are parameters. The space omits the mapping of the real $(x,y)$ 
coordinates to the space $[0,n)$, as well as the storage of the point values.
The 5th column gives simplified time assuming $\log n = \Theta(\log W)$, any
constant $\epsilon$, and use of $O(n\log n)$ bits.}
\label{tab:results}
\end{table}

\begin{table}[bt]
\begin{center}
\scalebox{0.65}{
\begin{tabular}{l@{~}|@{~}c@{~}|@{~}c@{~}|@{~}c} 
Operation & Space per point (bits) & Query time & Update time \\ 
\hline
$\somaN$, $\avgaN$, $\varaN$ 
	& $\log U(1+o(1)+1/t)$ 
		& $O(\log U \log n(1+\frac{\min(t\log W,\log U)t\log W}{\log\log n}))$ 
			& $O(\log U \log n)$  \\
$\minaN$, $\maxaN$ 
	& $\log U(1+o(1)+1/t)$ 
		& $O(\log U \log n(1+\frac{\min(t\log W,\log U)t\log W}{\log\log n}))$ 
			& $O(\log U(\log n + t\log W))$  \\
$\quantileN$ 
	& $\log U \log_\ell W (1+o(1))$ 
		& $O(\ell \log U \log n \log_\ell W/\log\log n)$ 
			& $O(\log U\log n \log_\ell W / \log\log n)$ \\
$\majaN$, var.
	& $\log U \log_\ell W (1+o(1))$ 
		& $O(\frac{1}{\alpha}\ell \log U\log n \log_\ell W/\log\log n)$ 
			& $O(\log U\log n \log_\ell W / \log\log n)$ \\
$\succaN$, $\predaN$
	& $\log U \log_\ell W (1+o(1))$ 
		& $O(\ell \log U \log n \log_\ell W/\log\log n)$ 
			& $O(\log U\log n \log_\ell W / \log\log n)$ \\
\end{tabular}
}
\end{center}
\caption{Our dynamic results on statistical queries, for $n$ two-dimensional 
points on an $U \times U$ grid with associated values in $[0,W)$; $2 \le \ell 
\le W$, $t \ge 1$ and $0<\epsilon<1$ are parameters. The space omits the 
mapping of the real $x$ coordinates to $[0,n)$, as well as the point values.
The first line is proved in Thm.~\ref{thm:dyn-soma}, the second in
Thm~\ref{thm:dyn-mina}, and the rest in Thm.~\ref{thm:dynamic}.}
\label{tab:dynresults}
\end{table}

In this case we build a wavelet tree on the universe of the point values. A 
sub-grid at each node stores the points whose values are within a range.
With this structure we can also solve mode and {\em top-$k$ most-frequent} 
queries.

Table~\ref{tab:results} shows the time and space results we obtain in this 
article for statistical queries. Several of our data structures can be made
dynamic at the price of a sublogarithmic penalty factor in the time 
complexities, as summarized in Table~\ref{tab:dynresults}.

\section{Wavelet Trees} \label{sec:basics}

Wavelet trees \cite{GGV03} are defined on top of the basic \rank\ and \select\
functions. Let $B$ denote a bitmap, \ie\ a sequence of $0$'s and
$1$'s. \frank{B}{b}{i} counts the number of times bit
$b \in \{0,1\}$ appears in $B[0,i]$, assuming $\frank{B}{b}{-1} = 0$. 
The dual operation, \fselect{B}{b}{i}, returns the position of the $i$-th
occurrence of $b$, assuming $\fselect{B}{b}{0}=-1$. 

The wavelet tree represents a sequence $S[0,n)$ over alphabet 
$\Sigma =[0,\sigma)$, and supports access to any $S[i]$, as well as 
\rank\ and \select\ on $S$, by 
reducing them to bitmaps. It is a complete binary tree where each node $v$ may
have a left child labeled $0$ (called the $0$-child of $v$) and a right child
labeled $1$ (called the 1-child). The sequence of labels obtained when
traversing the tree from the $\raiz$ down to a node $v$ is the {\em
  binary label} of $v$ and is denoted $L(v)$. Likewise we denote 
$V(L)$ the node that is obtained by following the sequence of bits
$L$, thus $V(L(v)) = v$. The binary labels of the leaves
correspond to the binary representation of the symbols of
$\Sigma$. Given $c \in \Sigma$ we denote by $V(c)$ 
the leaf that corresponds to symbol $c$. By
$c\{..d\}$ we denote the sequence of the first $d$ bits in
$c$. Therefore, for increasing values of $d$, the $V(c\{..d\})$ nodes
represent the path to $V(c)$.

Each node $v$ represents (but does not store) the subsequence $S(v)$
of $S$ formed by the symbols whose binary code starts with $L(v)$. At
each node $v$ we only store a (possibly empty) bitmap, denoted $B(v)$,
of length $|S(v)|$, so that $B(v)[i]=0$ iff $S(v)[i]\{..d\}=L(v)\cdot
0$, where $d=|L(v)|+1$, that is, if $S(v)[i]$ also belongs to the
$0$-child. A bit position $i$ in $B(v)$ can be mapped to a position in
each of its child nodes: we map $i$ to position $R(v,b,i) =
\frank{B(v)}{b}{i} - 1$ of the $b$-child.  We refer to this procedure
as the {\em reduction} of $i$, and use the same notation to represent
a sequence of steps, where $b$ is replaced by a sequence of bits. Thus
$R(\raiz, c, i)$, for a symbol $c \in \Sigma$, represents the
reduction of $i$ from the $\raiz$ using the bits in the binary
representation of $c$. With this notation we describe the way in
which the wavelet tree computes $\rank$, which is summarized by the
equation $ 
\frank{S}{c}{i} = R(\raiz, c, i)+1.  $ 
We use a similar notation $\Rp(v, v', i)$, to represent descending
from node $v$ towards a given node $v'$, instead of
explicitly describing the sequence of bits $b$ such that $L(v')= L(v)
\cdot b$ and writing $R(v, b, i)$.

An important path in the tree is obtained by choosing
$R(v,B(v)[i],i)$ at each node, \ie\ at each node we decide to go left
of right depending on the bit we are currently tracking. The resulting
leaf is $V(S[i])$, therefore this process provides a way to obtain the
elements of $S$. The resulting position is $R(\raiz, S[i], i) =
\rank(S,S[i],i)-1$.

It is also possible to move upwards on the tree, reverting the process
computed by $R$. Let node $v$ be the $b$-child of $v'$.  Then, if $i$
is a bit position in $B(v)$, we define the position $Z(v, v', i)$, in
$B(v')$, as $\fselect{B(v')}{b}{i+1}$. In general, when $v'$ is an
ancestor of $v$, the notation $Z(v, v', i)$ represents the iteration of
this process. For a general sequence, $\select$ can be computed by
this process, as summarized by the equation $ \fselect{S}{c}{i} =
Z(V(c), \raiz, i-1)$.

\begin{lemma}[\cite{GGV03,MN07,GGGRR07}]
  The wavelet tree for a sequence $S[0,n)$ over alphabet
  $\Sigma = [0,\sigma)$ requires at most $n\log\sigma + o(n)$
  bits of space.\footnote{From now on the space will be measured in bits and
$\log$ will be to the base 2.}
  It solves \rank, \select, and access to any $S[i]$ in time $O(\log\sigma)$.
\end{lemma}
\begin{proof}
  Grossi \ai~\cite{GGV03} proposed a representation using $n\log\sigma$ $+
  O(\frac{n\log\sigma\log\log n}{\log n})$ $+ O(\sigma\log n)$
  bits. M\"akinen and Navarro showed how to use only one pointer per
  level, reducing the last term to $O(\log\sigma\log n) = O(\log^2 n)
  = o(n)$.  Finally, Golynski \ai~\cite{GGGRR07} showed how to support
  binary \rank\ and \select\ in constant time, while reducing the
  redundancy of the bitmaps to $O(n\log\log n/\log^2 n)$, which added
  over the $n\log\sigma$ bits gives $o(n)$ as well.
 \end{proof}

\subsection{Representation of Grids}
\label{sec:rect-repr}

Consider a set $\Pe$ of $n$ distinct two-dimensional points $(x,y)$ from a
universe $[0,U) \times [0,U)$. We map coordinates to rank space using a 
standard method \cite{Cha88,ABR00}: We store two sorted arrays $X$ and $Y$ 
with all the (possibly repeated) $x$ and $y$ coordinates, respectively. Then 
we convert any point $(x,y)$ into rank space $[0,n) \times [0,n)$ in time 
$O(\log n)$ using two binary searches. Range queries are also mapped to rank 
space via binary searches (in an inclusive manner in case of repeated values). 
This mapping time will be dominated by other query times.

Therefore we store the points of $\Pe$ on a $[0,n) \times [0,n)$ grid, with 
exactly one point per row and one per column. We regard this set as a sequence 
$S[0,n)$ and the grid is formed by the points $(i,S[i])$. Then we represent 
$S$ using a wavelet tree.

The space of $X$ and $Y$ corresponds to the bare point data and will
not be further mentioned; we will only count the space to store the points in
rank space, as usual in the literature. In \ref{app:succgrids} we 
show how we can 
represent this mapping into rank space so that, together with a wavelet tree 
representation, the total space is only $O(n)$ bits over the minimum given by 
information theory.

The information relative to a point $p_0 = (x_0, y_0)$ is usually
tracked from the $\raiz$ and denoted $R(\raiz, y_0\{..d\}, x_0)$. A
pair of points $p_0 = (x_0, y_0)$ and $p_1 = (x_1, y_1)$, where $x_0
\le x_1$ and $y_0 \le y_1$, defines a rectangle; this is the typical
query range we consider in this paper. Rectangles have an implicit
representation in wavelet trees, spanning $O(\log n)$ nodes
\cite{MN07}. The binary representation of $y_0$ and $y_1$ share a
(possibly empty) common prefix. Therefore the paths $V(y_0\{..d\})$
and $V(y_1\{..d\})$ have a common initial path and then split at some
node of depth $k$, \ie\ $V(y_0\{..d\}) = V(y_1\{..d\})$ for $d \leq k$
and $V(y_0\{..d'\}) \neq V(y_1\{..d'\})$ for $d' > k$. Geometrically,
$V(y_0\{..k\})=V(y_1\{..k\})$ corresponds to the smallest horizontal
band of the form $[j\cdot n/2^k,(j+1)\cdot n/2^k)$ that contains the
query rectangle $Q$, for an integer $j$.  For $d' > k$ the nodes
$V(y_0\{..d'\})$ and $V(y_1\{..d'\})$ correspond respectively to
successively thinner, non-overlapping bands that contain the
coordinates $y_0$ and $y_1$.

Given a rectangle $Q = [x_0,x_1] \times [y_0,y_1]$ we consider the nodes
$V(y_0\{..d\}\cdot 1)$ such that $y_0\{..d\}\cdot 1 \neq y_0\{..d+1\}$, and the
nodes $V(y_1\{..d\}\cdot 0)$ such that $y_1\{..d\}\cdot 0
\neq y_1\{..d+1\}$. These nodes, together with $V(y_0)$ and $V(y_1)$,
form the implict representation of $[y_0,y_1]$, denoted 
$\imp$. The size of this set is $O(\log n)$. Let us recall a well-known
application of this decomposition.
\begin{lemma}
  Given $n$ two-dimensional points,
  the number of points inside a query rectangle $Q = [x_0,x_1] \times
  [y_0,y_1]$, $\conta$, can be computed in time $O(\log n)$
  with a structure that requires $n\log n +o(n)$ bits.
\label{lem:conta}
\end{lemma}
\begin{proof}
The result is 
$ 
\sum_{v \in \imp} \Rp(\raiz, v, x_1) - \Rp(\raiz, v, x_0-1).
$ 
Notice that all the values $R(\raiz, y_0\{..d\}, x)$ and $R(\raiz,
y_1\{..d\}, x)$ can be computed sequentially, in total time $O(\log n)$,
for $x=x_1$ and $x=x_0-1$. For a node $v \in \imp$ the desired difference can 
be computed from one of these values in time $O(1)$. Then the lemma follows.
 \end{proof}

This is not the best possible result for this problem (a better result 
by Bose et al.~\cite{BHMM09} exists), but it
is useful to illustrate how wavelet trees solve range search problems.

\section{Geometric Queries}
\label{sec:geometric-queries}
In this section we use wavelet trees to solve, within succinct space, two
geometric problems of relevance for data analysis.
\subsection{Dominating Points}
Given points $p_0 = (x_0, y_0)$ and $p_1 = (x_1, y_1)$, we say that
$p_0$ {\em dominates} $p_1$ if $x_0 \ge x_1$ and $y_0 \ge y_1$. Note that
one point can dominate the other even if they coincide in one coordinate.
Therefore, for technical convenience, in the reduction described in
Section~\ref{sec:rect-repr}, points with the same $y$ coordinates must be 
ranked in $Y$ by increasing $x$ value, and points with the same $x$ coordinates
must be ranked in $X$ by increasing $y$ value.
A point is {\em dominant} inside a range if there is no other point in
that range that dominates it. In Fig.~\ref{fig:wvtDiv} we define
the cardinal points N, S, E, W, SW, etc. We first use wavelet
trees to determine dominant points within rectangles.
\begin{theorem}
  Given $n$ two-dimensional points, the $d$ dominating points inside a
  rectangle $Q = [x_0,x_1] \times [y_0,y_1]$ can be obtained (in NW to
  SE order) in time $O((d+1)\log n)$, with a data structure using $n \log n +
  o(n)$ bits.
\label{thm:4-sided-range}
\end{theorem}
\begin{proof}
Let $v \in \imp$ be nodes in the implicit representation of $[y_0,y_1]$. We 
perform depth-first searches (DFS) rooted at each $v\in\imp$, starting from 
$V(y_1)$ and continuing sequentially to the left until
$V(y_0)$. Each such DFS is computed by first visiting the $1$-child and then 
the $0$-child. As a result, we will find the points in N to S order. 

We first describe a DFS that reports all the
nodes, and then restrict it to the dominant ones. Each
visited node $v'$ tracks the interval $(\Rp(\raiz, v', x_0-1),
\Rp(\raiz, v', x_1)]$. If the interval is empty we skip the subtree below
$v'$. As the grid contains only one point per row, for leaves $v'=v^{(i)}$ the 
intervals $(\Rp(\raiz, v^{(i)}, x_0-1), \Rp(\raiz, v^{(i)}, x_1)]$ contain at 
most one value, corresponding to a point $p^{(i)} \in \Pe\cap Q$. 
Then $x^{(i)} = Z(v^{(i)},\raiz,\Rp(\raiz, v^{(i)}, x_1))$ and $p^{(i)} =
(x^{(i)}, S[x^{(i)}])$. Reporting $k$ points in this way takes $O((k+1)
\log n)$ time.

By restricting the intervals associated with the nodes we obtain only
dominant points. 
In general, let $v'$ be the current node, that is either in $\imp$ or
it is a descendant of a node in $\imp$, and let $(x^{(i)},S[x^{(i)}])$
be the last point that was reported. Instead of considering the
interval $(\Rp(\raiz, v, x_0-1), \Rp(\raiz, v, x_1)]$, consider the
interval $(\Rp(\raiz, v, x^{(i)}), \Rp(\raiz, v, x_1)]$. This is correct
as it eliminates points $(x,y)$ with $x < x^{(i)}$, and also 
$y < S[x^{(i)}]$, given the N to S order in which we deliver the points.

As explained, a node with an empty interval is skipped.  On the other
hand, if the interval is non-empty, it must produce at least one
dominant point. Hence the cost of reporting the $d$ dominant points
amortizes to $O((d+1)\log n)$.
 \end{proof}

\begin{figure}[ht]
\begin{minipage}[b]{0.48\linewidth}
\centering
  \input{division-4.tex}
  \caption{Dominance on wavelet tree coordinates. The grayed points dominate
all the others in the rectangle. We also show the 4 directions.}
  \label{fig:wvtDiv}
\end{minipage}
\hspace{0.5cm}
\begin{minipage}[b]{0.48\linewidth}
\centering
  \input{division-1.tex}
  \caption{Rectangle visibility. For SW visibility the problem is the same as
dominance. We grayed the points that are visible in the other 3 directions.}
  \label{fig:wvtDiv2}
\end{minipage}
\end{figure}

\subsection{Rectangle Visibility}

Rectangle visibility is another, closely related, geometric problem. A
point $p \in \Pe$ is {\em visible} from a point $q = (x_0,y_0)$, not necessarily
in $\Pe$, if the rectangle defined by $p$ and $q$ as diagonally opposite
corners does not contain any point of $\Pe$. Depending on the direction from
$q$ to $p$ the visibility is called SW, SE, NW or NE (see 
Fig.~\ref{fig:wvtDiv2}). Next we solve visibility as variant of dominance.

\begin{theorem}
  The structure of Thm.~\ref{thm:4-sided-range} can compute the $d$ points 
  that are visible from a query point $q = (x_0, y_0)$, in order, in time 
  $O((d+1)\log n)$.
\label{thm:rectangle-visibility}
\end{theorem}
\begin{proof}
Note that SW visibility corresponds precisely to
determining the dominant points of the region $[0, x_0] \times
[0, y_0]$. Hence we use the procedure in
Thm.~\ref{thm:4-sided-range}. We now adapt it to the three
remaining directions without replicating the structure. 

For NE or SE visibility, we change the definition of 
operation $R$ to $R(v,b,i) = \frank{B(v)}{b}{i-1}+1$, thus
$\frank{S}{c}{i} = R(\raiz, c, i)+1$ if
$S[i] = c$ and $\frank{S}{c}{i} = R(\raiz, c, i)$ otherwise. In this
case we track the intervals $[\Rp(\raiz, v', x_0), \Rp(\raiz, v', x_1+1))$.
This $R(\raiz,v',x_1+1)$ is replaced by $R(\raiz,v',x^{(i)})$ when restricting
the search with the last dominant point.

For NE or NW visibility, the DFS searches first visit the $0$-child and then 
use the resulting points to restrict the search on the visit to
the 1-child, moreover they first visit the node $V(y_0)$ and
move to the right. 

Finally, for NW or SE visibility our point ordering in presence of ties
in $X$ or $Y$ may report points with the same $x$ or $y$ coordinate. To avoid
this we detect ties in $X$ or $Y$ at the time of reporting, right after
determining the pair $p^{(i)} = (x,y) = (x^{(i)},S[x^{(i)}])$. In the NW (SE)
case, we binary search for the last (first) positions $x'$ such that $X[x'] = 
X[x]$ and $y'$ such that $Y[y'] = Y[y]$. Then we correct $p^{(i)}$ to $(x',y)$
(to $(x,y')$). The subsequent searches are then limited by $x'$ instead of 
$x=x^{(i)}$. We also limit subsequent searches in a new way: we skip traversing
subsequent subtrees of $\imp$ until the $y$ values are larger (NW) or smaller 
(SE) than $y'$. Still the cost per reported point is $O(\log n)$.
 \end{proof}

\subsection{Dynamism}

 We can support point insertions and deletions on a fixed $U \times U$ grid.
 Dynamic variants of the bitmaps stored at each wavelet tree node raise the
 extra space to $o(\log U)$ per point and multiply the times 
 by $O(\log n /\log\log n)$ \cite{HM10,NS10}. 

 \begin{lemma}
   Given $n$ points on a $U \times U$ grid, there is a structure using
   $n\log U + o(n\log U)$ bits, answering 
   queries in time $O(t(\log U)\log n / \log\log n)$, where $t(h)$ is the time 
   complexity of the query using static wavelet trees of height $h$. It handles
   insertions and deletions in time $O(\log U \log n / \log\log n)$.
   \label{lem:dyn-dom}
 \end{lemma}
 \begin{proof}
 We use the same data structure and query algorithms of 
the static wavelet trees described in Section~\ref{sec:rect-repr},
 yet representing their bitmaps with the dynamic variants \cite{HM10,NS10}. 
 We also maintain vector $X$, but not $Y$; we use the $y$-coordinates 
 directly instead since the wavelet tree handles repetitions in $y$.
Having a wavelet tree of depth $\log U$ makes the time $t(\log U)$, whereas
using dynamic bitmaps multiplies this time by $O(\log n / \log\log n)$, 

Instead of an array, we use for $X$ a B-tree tree with arity $\Theta(\log U)$.
Nodes are handled with a standard technique for managing cells of different 
sizes \cite{Mun86}, which wastes just $O(\log^2 U)$ bits in total. As a result,
the time for accessing a position of $X$ or for finding the range of elements 
corresponding to a range of coordinates is $O(\log U)$, which is subsumed by 
other complexities. The extra space on top of that of the bare coordinates is 
$O(n + \log^2 U)$ bits. This is $o(n\log U)$ unless $n = o(\log U)$, in which
case we can just store the points in plain form and solve all queries 
sequentially. It is also easy to store differentially encoded coordinates in
this B-tree to reduce the space of mapping the universe of $X$ coordinates
to the same achieved in Section~\ref{sec:rect-repr}.

When inserting a new point $(x,y)$, apart from inserting $x$ into $X$, we 
track the point downwards in the wavelet tree, doing the insertion at each of 
the $\log U$ bitmaps. Deletion is analogous.
 \end{proof}

As a direct application, dominance and visibility queries can be solved in the
dynamic setting in time $O((d+1)\log U \log n / \log\log n)$, while supporting
point insertions and deletions in time $O(\log U \log n / \log\log n)$. 
The only issue is that we now may have several points with the same 
$y$-coordinate, which is noted when the interval is of size more than one 
upon reaching a leaf. In this case, as these points are sorted by increasing 
$x$-coordinate, we only report the first one (E) or the last one (W). Ties in 
the $x$-coordinate, instead, are handled as in the static case.

\section{Range Sum, Average, and Variance}
\label{sec:sum}

We now consider points with an associated value given by an integer function 
$w:\Pe \rightarrow [0,W)$. 
We define the sequence of values $W(v)$ associated to each wavelet tree node 
$v$ as follows: If $S(v) =
p_0,p_1,\ldots, p_{|S(v)|}$, then $W(v) =w(p_0),w(p_1),\ldots,w(p_{|S(v)|})$.
%
%
We start with a solution to several range sum problems on groups.
In our results we will omit the bare $n \lceil \log W \rceil$ bits needed to
store the values of the points.

\begin{theorem}
  Given $n$ two-dimensional points with associated values in $[0,W)$,
  the sum of the point values inside a query rectangle $Q = [x_0,x_1] \times
  [y_0,y_1]$, $\soma$, can be computed in time 
  $O(\min(t\log W,\log n) t \log W \log n)$,
  with a structure that requires $n\log n (1+ 1/t)$
  bits, for any $t\ge 1$. It can also compute the average and
  variance of the values, $\avga$ and $\vara$ respectively.
  \label{thm:soma}
\end{theorem}
\begin{proof}
We enrich the bitmaps of the wavelet tree for $\Pe$. For each node $v$ we 
represent its vector $W(v) = w(p_0),w(p_1),\ldots,w(p_{|S(v)|})$ as a 
bitmap $A(v)$, where we concatenate the unary representation of the
$w(p_i)$'s, \ie\ $w(p_i)$ $0$'s followed by a $1$. These bitmaps $A(v)$ 
are represented in a compressed format~\cite{OS07} that requires at most $|S(v)| \log W + O(|S(v)|)$
bits. With this structure we can determine the sum $w(p_0) + w(p_1) +
\ldots +  w(p_i)$, \ie\ the partial sums, in constant time by means of 
$\fselect{A(v)}{1}{i}$ queries\footnote{Using constant-time $\select$ 
structures on their internal bitmap $H$ \cite{OS07}.}, $\wsum(v,i) = \,
$\fselect{A(v)}{1}{i+1}$-i$ is the sum of the first $i+1$ values.
In order to compute $\soma$ we use a formula similar to the one of
Lemma~\ref{lem:conta}:
\begin{equation}
 \sum_{v \in \imp} \wsum(v,\Rp(\raiz, v, x_1)) - \wsum(v,\Rp(\raiz, v, x_0-1)).
\label{eq:wsum}
\end{equation}

To obtain the tradeoff related to $t$, we call $\tau = t \log W$ and store 
only every $\tau$-th entry in 
$A$, that is, we store partial sums only at the end of {\em blocks} of $\tau$
entries of $W(v)$. We lose our ability to compute $\wsum(v,i)$ exactly, but 
can only compute it for $i$ values that are at the end of blocks,
$\wsum(v,\tau\cdot i) = \fselect{A(v)}{1}{i+1}-i$. To
compute each of the terms in the sum of Eq.~(\ref{eq:wsum}) we can use 
$\wsum(v,\tau\cdot\lfloor \Rp(\raiz, v, x_1)/\tau\rfloor) - 
 \wsum(v,\tau\cdot\lceil\Rp(\raiz, v, x_0-1)/\tau\rceil)$
to find the sum of the part of the range that covers whole blocks. Then we 
must find out the remaining (at most) $2\tau-2$ values $w(p_i)$ that lie at 
both extremes of the range, to complete the sum.

In order to find out those values, we store the vectors $W(v)$ explicitly
at all the tree nodes $v$ whose height $h(v)$ is a multiple of 
$\tau$, including the leaves. If a node $v \in \imp$ 
does not have stored its vector $W(v)$, it can
still compute any $w(p_i)$ value by tracking it down for at most $\tau$ levels.

As a result, the time to compute a $\soma$ query is $O(\tau^2\log n)$, yet it is
limited to $O(\tau \log^2 n)$ if $\tau > \log n$, as at worst we have $W(v)$
represented at the leaves.
The space for the $A(v)$ vectors is at most $(|S(v)|/\tau)(\log W + O(1))$ bits, 
which adds up to $n \log n (\log W+O(1))/\tau$ bits. On the other hand, the 
$W(v)$ vectors add up to $n \log n (\log W)/\tau = n(\log n) / t$ bits. 
This holds for any $t$
even when we store always the $W(v)$ vectors at the leaves: The space of those
$W(v)$ is not counted because we take for free the space needed to represent 
all the values once, as explained.


The average inside $Q$ is computed as $\avga = \soma
/\conta$, where the latter is computed with the same structure by just adding
up the interval lengths in $\imp$. To compute variance we use, conceptually, 
an additional instance of the same data structure, with values 
$w'(p) = w^2(p)$. Then $\vara = \somap/\conta - (\soma/\conta)^2$, where
$\somapN$ uses the values $w'$. Note that in fact we only need to store
additional (sampled) bitmaps $A'(v)$ corresponding to the partial sums of
vectors $W'(v)$ (these bitmaps may need twice the space of the $A(v)$ bitmaps
as they handle values that fit in $2\log W$ bits). Explicitly stored vectors 
$W'(v)$ are not necessary as they can be emulated with $W(v)$, and we can also 
share the same wavelet tree structures and bitmaps. This extra space fits
within the same $O(n (\log n)/t)$ bits.
 \end{proof}

\ref{app:variance} shows how to 
 further reduce the constant hidden in the $O$ notation. 
This is important because this constant is also associated with the
 $\lceil \log W \rceil$ bits of the weights, that are being
 omitted from the analysis: In the case of $w'$ we have $2 \lceil \log
 W \rceil$ bits per point. 

\paragraph{\bf\em Finite groups}
The solution applies to finite groups $(G,\oplus,^{-1},0)$.
We store $\wsum(v,i) = w(p_0) \oplus w(p_1) \oplus \ldots w(p_i)$,
directly using $\lceil \log |G| \rceil$ bits per entry. The terms
$\wsum(v,i)-\wsum(v,j)$ of Eq.~(\ref{eq:wsum}) are replaced by
$\wsum(v,j)^{-1} \oplus \wsum(v,i)$.

\subsection{Dynamism}
\label{sec:dynsum}

A dynamic variant is obtained by using the dynamic wavelet trees of
Lemma~\ref{lem:dyn-dom}, a dynamic partial sums data structure
instead of $A(v)$, and a dynamic array for vectors $W(v)$.

\begin{theorem}
  Given $n$ points on a $U \times U$ grid, with associated values in $[0,W)$,
  there is a structure that uses $n\log U (1 + o(1) + 1/t)$ bits, 
  for any $t\ge 1$, that answers the queries in Theorem~\ref{thm:soma} in time 
  $O(\log U \log n(1+\min(t\log W,\log U)t\log W$ $/\log\log n))$, and supports
  point insertions/deletions, and value updates, in time $O(\log U\log n)$.
  \label{thm:dyn-soma}
\end{theorem}
\begin{proof}
The algorithms on the wavelet tree bitmaps are carried out verbatim,
now on the dynamic data structures of Lemma~\ref{lem:dyn-dom}, which add
$o(n\log U)$ bits of space overhead and multiply the times by $O(\log U /
\log\log n)$. This adds $O(\min(t\log W,\log U)t\log W\log U\log n/\log\log n)$ 
to the query times.

Dynamic arrays to hold the explicit $W(v)$ vectors can be implemented within 
$|S(v)|\log W(1+o(1))$ bits, and provide access and indels in time 
$O(\log n / \log\log n)$ \cite[Lemma~1]{NS10}. This adds 
$O(t \log W \log n / \log\log n)$ to the query times, which is negligible.

For insertions we must insert the new bits at all the levels as in
Lemma~\ref{lem:dyn-dom}, which costs $O(\log U \log n /\log\log n)$
time, and also insert the new values in $W(v)$ at $1+(\log U)/(t\log
W)$ levels, which in turn costs time $O((1+(\log U)/(t\log W))$ $\log n /
\log\log n)$ (this is negligible compared to the cost of updating
the bitmaps). Deletions are analogous.  To update a value we just
delete and reinsert the point.

A structure for dynamic searchable partial sums \cite{MN08} takes
$n\log W + o(n\log W)$ bits to store an array of $n$ values, and supports 
partial sums, as well as insertions and deletions of values, in time 
$O(\log n)$. 
Note that we carry out $O(\log U)$ partial sum operations per query.
We also perform $O(\log U)$ updates when points are inserted/deleted.
This adds $O(\log U \log n)$ time to both query and update complexities.

Maintaining the {\em sampled} partial sums $A(v)$ is the most complicated part. 
Upon insertions and deletions we cannot maintain a fixed block size $\tau$. 
Rather, we use a technique \cite{GN08} that ensures that the blocks are
of length at most $2\tau$ and two consecutive blocks add up to at least
$\tau$. This is sufficient to ensure that our space and time complexities hold.
The technique does not split or merge blocks,
but it just creates/removes empty blocks, and moves one
value to a neighboring block. All those operations are easily carried out with
a constant number of operations in the dynamic partial sums data structure.  

Finally, we need to mark the positions where the blocks start. We can maintain
the sequence of $O(|S(v)|/\tau)$ block lengths using again a partial sums data
structure \cite{MN08}, which takes $O((|S(v)|/\tau) \log \tau)$ bits. 
The starting position of any block is obtained as a partial sum, in time 
$O(\log n)$, and the updates required when blocks are created or 
change size are also carried out in time $O(\log n)$. These are 
all within the same complexities of the partial sum structure for
$A(v)$.
\end{proof}

\paragraph{\bf\em Finite groups and semigroups}

The solution applies to finite groups $(G,\oplus,^{-1},$ $0)$. The dynamic 
structure for partial sums \cite{MN08} can be easily converted into one
that stores the local ``sum'' $w(p_i) \oplus w(p_{i+1}) \oplus \ldots w(p_j)$
of each subtree containing leaves $p_i, p_{i+1}, \ldots p_j$.
The only obstacle in applying it to semigroups is that we cannot move
an element from one block to another in constant time, because we have to
recalculate the ``sum'' of a block without the element removed. This takes
time $O(\tau)$, so the update time becomes $O(\log U (\log n + t\log W))$.

\section{Range Minima and Maxima}
\label{sec:rmq}

For the one-dimensional problem there exists a data structure using just
$2n+o(n)$ bits, which answers queries in constant time without accessing
the values \cite{Fis10}. This structure allows for a better space/time
tradeoff compared to range sums.

For the queries that follow we do not need the exact $w(p)$ values, but just
their relative order. So we set up a bitmap $V[1,W]$ where the values
occurring in the set are marked. This bitmap can be stored within at most
$m\log(W/m) + O(m)$ bits \cite{OS07}, where $m \le \min(n,W)$ is the number of 
unique values. This representation converts between actual value
and relative order in time $O(\log m)$, which will be negligible.
This way, many complexities will be expressed in terms of $m$.

\begin{theorem}
  Given $n$ two-dimensional points with associated values in $[0,W)$,
  the minimum of the point values inside a query rectangle $Q = [x_0,x_1] \times
  [y_0,y_1]$, $\mina$, can be found in time 
  $O(\min(t\log m,\log n)t\log n)$,
  with a structure using $n\log n (1+1/t)$ bits, for any $t \ge 1$ 
  and $m=\min(n,W)$. 
  The maximum of the point values inside a query rectangle $Q$ can be found
  within the same time and space bounds.
  \label{thm:mina}
\end{theorem}
\begin{proof}
We associate to each node $v$ the one-dimensional data structure \cite{Fis10}
corresponding to $W(v)$, which takes $2|W(v)|+o(|W(v)|)$ bits. This adds 
up to $2n\log n + o(n\log n)$ bits overall. We call 
$\wmin(v,i,j) = \mathrm{arg}\,\min_{i\le s\le j} W(v)[s]$ the one-dimensional
operation. Then we can find in constant time the position of the minimum value 
inside each $v \in \imp$ (without the need to store the values in the node), 
and the range minimum is
\begin{equation*}
 \min_{v \in \imp} W(v)[\wmin(v,\Rp(\raiz, v, x_0),\Rp(v,\raiz, v, x_1+1)-1)].
\end{equation*}
To complete the comparison we need to compute the $O(\log n)$ values 
$W(v)[s]$ of different nodes $v$. By storing the $W(v)$ vectors
of Theorem~\ref{thm:soma} (in the range $[1,m]$) every $\tau=t\log m$ levels, 
the time is just $O(\min(\tau,\log n) \log n)$ because we have to track down just 
one point for each $v \in \imp$. The space is $3n\log n + (n\log n \log m)/\tau
= 3n\log n + (n\log n)/t$ bits. The second term holds for any $t$ even when we 
always store $n\log m$ bits at the leaves, because adding these to the 
$m\log(W/m)+O(m)$ bits used for $V$, we have the $n \lceil \log W \rceil$ bits 
corresponding to storing the bare values and that are not accounted for in our 
space complexities.

To reduce the space further, we split $W(v)$ into blocks of length $r$ and
create a sequence $W'(v)$, of length $|S(v)|/r$, where we take the minimum
of each block, $W'(v)[i] = \min \{ W(v)[(i-1)\cdot r+1], \ldots, 
W(v)[i\cdot r]\}$. The one-dimensional data structures are built over $W'(v)$,
not $W(v)$, and thus the overall space for these is $O((n/r)\log n)$ bits. In
exchange, to complete the query we need to find the $r$ values covered by
the block where the minimum was found, plus up to $2r-2$ values in the extremes
of the range that are not covered by full blocks. The time is thus
$O(r\min(\tau,\log n)\log n)$. By setting $r = t$ we obtain the result.

For $\maxa$ we use analogous data structures.
 \end{proof}

\paragraph{\bf\em Top-$k$ queries in succinct space}

We can now solve the top-$k$ query of Rahul \ai~\cite{RGJR11} by iterating
over Theorem~\ref{thm:mina}. Let us set $r=1$. Once we identify that the overall
minimum is some $W(v)[s]$ from the range $W(v)[i,j]$, we can find the second 
minimum among the other candidate ranges plus the ranges $W(v)[i,s-1]$ and 
$W(v)[s+1,j]$. As this is repeated $k$ times, we pay time $O(\tau(k+\log n))$ to
find all the minima. A priority queue handling the ranges will perform $k$ 
minimum extractions and $O(k+\log n)$ insertions, and its size will be limited 
to $k$. So the overall time is $O(\tau \log n + k(\tau + \log k))$ by using a 
priority queue with constant insertion time \cite{CMP88}. Using $\tau=t\log m$ 
for any $t=\omega(1)$ we obtain time $O(t\log m\log n + kt\log m\log k)$ and 
$n\log n+o(n\log n)$ bits of space. The best current linear-space solution 
\cite{NN12} achieves better time and linear space, but the constant multiplying
the linear space is far from 1.

\subsection{Dynamism} 

We can directly apply the result on semigroups given in 
Section~\ref{sec:dynsum}. Note that, while in the static scenario we 
achieve a better result than for sums, in the dynamic case the result is 
slightly worse.

\begin{theorem}
  Given $n$ points on an $U \times U$ grid, with associated values in $[0,W)$,
  there is a structure using $n\log U (1 + o(1) + 1/t)$ bits, 
  for any $t\ge 1$, that answers the queries in Theorem~\ref{thm:mina} in time 
  $O(\log U \log n(1+\min(t\log W,\log U)t\log W/\log\log n))$, and supports
  point insertions/deletions, and value updates, in time 
  $O(\log U (\log n + t\log W))$.
  \label{thm:dyn-mina}
\end{theorem}

\section{Range Majority for Fixed $\alpha$}
\label{sec:majf}

In this section we describe a data structure that answers 
$\alpha$-majority queries for the case where $\alpha$ is fixed at
construction time. Again, we enrich the wavelet tree with additional
information that is sparsified. We obtain the following result.

\begin{theorem}
  Given $n$ two-dimensional points with associated values in $[0,W)$ and
  a fixed value $0<\alpha<1$, all the $\alpha$-majorities inside a query 
  rectangle $Q = [x_0,x_1] \times [y_0,y_1]$, $\maja$, can be found in time 
  $O(t\log m\log^2 n)$, with a structure using $n((2+1/t)\log n + \log m)$
  bits, for any $t \ge 1$ and $m=\min(n,W)$. 
  \label{thm:fmaja}
\end{theorem}
\begin{proof}
We say that a set of values  $C$ is a set of {\em $\alpha$-candidates}
for $S'\subset S$ if each $\alpha$-majority value $w$ of $S'$ belongs to $C$.
In every wavelet tree node $v$ we store an auxiliary data structure $A(v)$ that 
corresponds to elements of $W(v)$. The data structure $A(v)$ enables us to 
find the set of $\alpha$-candidates for any range $[r_1\cdot s,r_2\cdot s]$ in $W(v)$,
for a parameter $s = t\log m$. We implement $A(v)$ as a balanced binary range 
tree $T(v)$ on $W(v)$. Every leaf of $T(v)$ corresponds to an interval of 
$W(v)$ of size $s/\alpha$. The range of an internal node $w$ of $T$ is the 
union of ranges associated with its children. In every such node $w$, we store 
all $\alpha$-majority values for the range of $w$ (these are at most $1/\alpha$
values). The space required by $T(v)$ is $(2|W(v)|/(s/\alpha))(1/\alpha)\log m 
= O(|W(v)|/t)$ bits, which added over all the wavelet tree nodes $v$ sums 
to $(n\log n)/t$ bits.

Given an interval $[r_1\cdot t,r_2\cdot t]$, we can represent it as a 
union of $O(\log n)$ ranges for nodes $w_i\in T(v)$. 
If a value is an $\alpha$-majority value for $[r_1\cdot t,r_2\cdot t]$, then it is
an $\alpha$-majority value for at least one $w_i$. 
Hence, a candidate set for $[r_1\cdot t,r_2\cdot t]$ is the union of values stored in 
$w_i$. The candidate set contains $O((1/\alpha)\log n)$ values and can be 
found in $O((1/\alpha)\log n)$ time. 

Moreover, for every value $c$, we store the grid $G(c)$ of 
Lemma~\ref{lem:conta}, which enables us to find the total number of elements 
with value $w$ in any range $Q=[x_0,x_1]\times [y_0,x_1]$.  
Each $G(c)$, however, needs to have the coordinates mapped to the rows
and columns that contain elements with value $c$. We store sequences
$X_c$ and $Y_c$ giving the value identifiers of the points sorted by $x$-
and $y$-coordinates, respectively. By representing them as wavelet trees,
they take $2n\log m + o(n)$ bits of space and map in constant time any
range $[x_0,x_1]$ or $[y_0,y_1]$ using $rank$ operations on the sequences,
in $O(\log m)$ time using the wavelet trees. Then the local grids, which
overall occupy other $\sum_{c \in [1,m]} n_c \log n_c + o(n_c) \le 
n \log (n/m)+o(n)$ bits, complete the range counting query in time $O(\log n)$.
So the total space of these range counting structures is $n\log n + n\log m
+o(n)$.

To solve an $\alpha$-majority query in a range $Q=[x_1,x_2]\times [y_1,y_2]$,
we visit each node $v\in \imp$. We identify the longest interval 
$[r_1\cdot s,r_2\cdot s]\subseteq [\Rp(\raiz, v, x_0),\Rp(\raiz, v,x_1)]$. 
Using $A(v)$ the candidate values in $[r_1\cdot s,r_2\cdot s]$ can be found 
in time $O((1/\alpha)\log n)$. Then we obtain the values of the elements in 
$[\Rp(\raiz, v, x_0),r_1 \cdot s)$ and $(r_2\cdot s,\Rp(\raiz, v, x_1)]$, in time 
$O(s \log n)$ by traversing the wavelet tree. Added over all
the $v\in\imp$, the cost to find the $(1/\alpha + s)\log n$ candidates 
is $O((1/\alpha + s\log n)\log n)$. Then their frequencies
in $Q$ are counted using the grids $G(c)$ in time
$O((1/\alpha + s)\log^2 n)$, and the $\alpha$-majorities are finally
identified. 

Thus the overall time is $O(t \log m \log^2 n)$. The space is
$n(2\log n + \log m + (\log n)/t)$, higher than for the previous problems
but less than the structures to come.
\end{proof}

A slightly better (but messier) time complexity can be obtained by
using the counting structure of Bose \ai~\cite{BHMM09} instead of that
of Lemma~\ref{lem:conta}, storing value identifiers every $s$ tree
levels, $O(t\log m\log n (\min(t\log m,\log n)+\log n / \log\log n))$.
The space increases by $o(n\log (n/m))$. On the other hand, by using
$s=t=1$ we increase the space to $O(n\log n)$ integers and reduce the
query time to $O(\log^2 n)$.

\section{Range Median and Quantiles}
\label{sec:quantile}

We compute the median element, or more generally,
the $k$-th smallest value $w(p)$ in an area $Q = [x_0,x_1] \times [y_0,y_1]$
(the median corresponds to $k=\conta/2$). 

From now on we use a different wavelet tree decomposition, on the 
universe $[0,m)$ of $w(\cdot)$ values rather than on $y$ coordinates.
This can be seen as a wavelet tree on grids rather than on sequences:
the node $v$ of height $h(v)$ stores a grid $G(v)$ with the points $p \in \Pe$ 
such that $\lfloor w(p)/2^{h(v)} \rfloor = L(v\{..\lceil\log m\rceil-h(v)\})$. 
Note that each leaf $c$ stores the points $p$ with value $w(p)=c$.
\begin{theorem}
  Given $n$ two-dimensional points with associated values in $[0,W)$, 
  the $k$-th smallest value of points within a query rectangle $Q = [x_0,x_1] \times
  [y_0,y_1]$, $\quantile$, can be found in time $O(\ell \log n \log_\ell m)$,
  with a structure using $n\log n \log_\ell m+O(n\log m)$
  bits, for any $\ell\in[2,m]$ and $m=\min(n,W)$. 
\label{thm:quantile}
\end{theorem}
\begin{proof}
We use the wavelet tree on grids just described, representing each grid $G(v)$ 
with the structure of Lemma~\ref{lem:conta}. To solve this query we start 
at root of the wavelet tree of grids and consider its left child, $v$. If
$t = \conta \ge k$ on grid $G(v)$, we continue the search on $v$. Otherwise we
continue the search on the right child of the root, with parameter $k-t$. When
we arrive at a leaf corresponding to value $c$, then $c$ is the $k$-th smallest
value in $\Pe \cap Q$.

Notice that we need to reduce the query rectangle to each of the grids $G(v)$ 
found in the way. We store the $X$ and $Y$ arrays only for the root grid, 
which contains the whole $\Pe$. For this and each other grid $G(v)$, we store 
a bitmap $X(v)$ so that $X(v)[i] = b$
iff the $i$-th point in $x$-order is stored at the $b$-child of $v$.
Similarly, we store a bitmap $Y(v)$ with the same bits in $y$-order.
Therefore, when we descend to the $b$-child of $v$, for $b\in\{0,1\}$, we
remap $x_0$ to $\frank{X(v)}{b}{x_0}$ and $x_1$ to $\frank{X(v)}{b}{x_1+1}-1$,
and analogously for $y_0$ and $y_1$ with $Y(v)$.

The bitmaps $X(v)$ and $Y(v)$ add up to $O(n\log m)$ bits of space. For the
grids, consider that each point in each grid contributes at most $\log n
+o(1)$ bits, and each $p \in \Pe$ appears in $\lceil \log m \rceil-1$
grids (as the root grid is not really necessary). 

To reduce space, we store the grids $G(v)$ only every $\lceil \log \ell \rceil$
levels (the bitmaps $X(v)$ and $Y(v)$ are still stored for all the levels).
This gives the promised space. For the time, the first decision
on the root requires computing up to $\ell$ operations $\conta$, but this
gives sufficient information to directly descend $\log \ell$ levels.
Thus total time adds up to $O(\ell \log n \log_\ell m)$.
 \end{proof}

Again, by replacing our structure of Lemma~\ref{lem:conta} by Bose \ai's 
counting structure \cite{BHMM09}, the time drops to 
$O(\ell\log n\log_\ell m/\log\log n)$ when using 
$n\log n \log_\ell m (1+o(1)) + O(n\log m)$ bits of space.

The basic wavelet tree structure allows us to count the number of 
points $p \in Q$ whose values $w(p)$ fall in a given range $[w_0,w_1]$, within
time $O(\ell \log n \log_\ell m)$ or 
$O(\ell\log n\log_\ell m/\log\log n)$. This is another useful operation
for data analysis, and can be obtained with the formula
 $ 
  \sum_{v \in \impN(w_0,w_1)} \conta.
 $ 

As a curiosity, we have tried, just as done in Sections~\ref{sec:sum} and
\ref{sec:rmq}, to build a wavelet tree on the $y$-coordinates and use a
one-dimensional data structure. We used the optimal linear-space structure of
Brodal and J{\o}rgensen \cite{BJ09}. However, the result is not competitive
with the one we have achieved by building a wavelet tree on the domain of
point values.

\section{Range Majority for Variable $\alpha$}
\label{sec:majv}

We can solve this problem, where $\alpha$ is specified at query time, with the 
same structure used for Theorem~\ref{thm:quantile}.

\begin{theorem}
  The structures of Theorem~\ref{thm:quantile} can compute all the 
  $\alpha$-majorities of the point values inside $Q$, \maja, in time
  $O(\frac{1}{\alpha}\ell \log n \log_\ell m)$,
  where $\alpha$ can be chosen at query time.
\label{thm:maja}
\end{theorem}
\begin{proof}
For $\alpha \ge\frac{1}{2}$ we find the median $c$ of $Q$ and 
then use the leaf $c$ to count its frequency in $Q$. If this is more than 
$\alpha\cdot\conta$, then $c$ is the answer, else there is no 
$\alpha$-majority. For $\alpha < \frac{1}{2}$,
we solve the query by probing all the $(i\cdot\alpha)\conta$-th elements in 
$Q$. 
 \end{proof}

Once again, we attempted to build a wavelet tree on $y$-coordinates, using
the one-dimensional structure of Durocher et al.~\cite{DHMNS11} at each level,
but we obtain inferior results.

Culpepper \ai~\cite{CNPT10} show how to find the mode, and in general the $k$ 
most repeated values inside $Q$, using successively more refined $\quantileN$
queries. Let the $k$-th most repeated value occur $\alpha \cdot \conta$ times
in $Q$, then we require at most $4/\alpha$ quantile queries \cite{CNPT10}. The 
same result can be obtained by probing successive values $\alpha = 1/2^i$ with
$\majaN$ queries.

\section{Range Successor and Predecessor}
\label{sec:suc}

The successor (predecessor) of a value $w$ in a rectangle 
$Q=[x_0,x_1]\times[y_0,y_1]$ is the smallest (largest) value larger 
(smaller) than, or equal to, $w$ in $Q$. We also have an efficient solution
using our wavelet trees on grids.

\begin{theorem}
  The structures of Theorem~\ref{thm:quantile} can compute the successor and
  predecessor of a value $w$ within the values of the points inside $Q$,
  $\succa$ and \preda, in time $O(\ell \log n \log_\ell m)$.
  \label{thm:preda}
\end{theorem}
\begin{proof}
  We consider the nodes $v \in \impN(w, +\infty)$ from left to right, 
tracking rectangle $Q$ in the process.
  The condition for continuing the search below a node $v$ that is in 
  $\impN(w, +\infty)$, or is a descendant of one such node, is that
  $\contaN(Q) > 0$ on $G(v)$. $\succa$ is the value associated with the first
  leaf found by this process. Likewise, $\preda$ is computed by
  searching $\impN(-\infty, w)$ from right to left.
To reduce space we store the grids only every $\lceil \log \ell \rceil$
levels, and thus determining whether a child has a point in $Q$ may cost up to 
$O(\ell \log n)$. Yet, as for Theorem~\ref{thm:quantile}, the total time 
amortizes to $O(\ell \log n \log_\ell m)$. 
 \end{proof}

Once again, storing one-dimensional data structures \cite{CIKRW08,MNU05} on
a $y$-coordinate-based wavelet tree does not yield competitive results.

\section{Dynamism}
\label{sec:dynamism}

Our dynamic wavelet tree of Lemma~\ref{lem:dyn-dom} supports range counting 
and point insertions/deletions on a fixed grid in time 
$O(\log U \log n / \log\log n)$ (other tradeoffs exist \cite{Nek09}).
If we likewise assume that our grid is fixed in Theorems~\ref{thm:quantile},
\ref{thm:maja} and \ref{thm:preda}, we can also support point insertions
and deletions (and thus changing the value of a point). 

\begin{theorem}
  Given $n$ points on a $U \times U$ grid, with associated values in $[0,W)$,
  there is a structure using
  $n\log U \log_\ell W (1+o(1))$ bits, for any $\ell \in [2,W]$, that answers 
  the queries $\quantileN$, $\succaN$ and $\predaN$ in time
  $O(\ell \log U \log n \log_\ell W/$ $\log\log n)$, and the $\majaN$
  operations in time
  $O(\frac{1}{\alpha} \ell \log U \log n \log_\ell W /$ $ \log\log n)$. It supports
  point insertions and deletions, and value updates, in
  time $O(\log U \log n \log_\ell W / \log\log n)$. 
  \label{thm:dynamic}
\end{theorem}
\begin{proof}
We use the data structure of Theorems~\ref{thm:quantile}, \ref{thm:maja}
and \ref{thm:preda}, modified as follows. We build the wavelet tree on the
universe $[0,W)$ and thus do not map the universe values to rank space. The
grids $G(v)$ use the dynamic structure of Lemma~\ref{lem:dyn-dom}, on global
$y$-coordinates $[0,U)$. We maintain the global array $X$ of
Lemma~\ref{lem:dyn-dom} plus the vectors $X(v)$ of Theorem~\ref{thm:quantile},
the latter using dynamic bitmaps
\cite{HM10,NS10}. The time for the queries follows immediately.
For updates we track down the point to insert/delete across the wavelet 
tree, inserting or deleting it in each grid $G(v)$ found in the way, and
also in the corresponding vector $X(v)$.
\end{proof}

\section{Conclusions}
\label{sec:conclusions}

We have demonstrated how wavelet trees \cite{GGV03} can be used for solving
a wide range of two-dimensional queries that are useful for various data
analysis activities. Wavelet trees have the virtue of using little space. By
enriching them with further sparsified data, we support various complex queries
in polylogarithmic time and linear space, sometimes even succinct. Other 
more complicated queries require slightly superlinear space.

We believe this work just opens the door to the possible applications
to data analysis, and that many other queries may be of interest. A prominent 
one lacking good solutions is to find the mode, that is, the most frequent 
value, in a rectangle, and its generalization to the top-$k$ most frequent 
values. There has been some recent progress on the one-dimensional version 
\cite{GNP10} and even in two dimensions \cite{DM11}, but the results are far 
from satisfactory.

Another interesting open problem is how to support dynamism while retaining
time complexities logarithmic in the number of points and not in the grid
size. This is related to the problem of dynamic wavelet trees, in particular
supporting insertion and deletion of $y$-coordinates (on which they build the
partition). Dynamic wavelet trees would also solve many problems in other
areas.

Finally, a natural question is which are the lower bounds that relate the
achievable space and time complexities for the data analysis queries we have
considered. These are well known for the more typical counting and reporting
queries, but not for these less understood ones.

\section*{References}

\bibliographystyle{alpha}
\bibliography{paper}

\appendix

\section{Optimal-Space Representation of Grids} \label{app:succgrids}

We analyze the representation described in Section~\ref{sec:rect-repr}, showing
how it can be made near-optimal in the information-theoretic sense. Recall that
our representation of a set of points of $[0, U)^2$ consists in storing two 
sorted arrays $X$ and $Y$, which reduce the $[0,U)$ values to $[0,n)$. The 
points in the $[0,n) \times [0,n)$ grid have exactly one point per row and one
per column.

An optimal-space representation of the above data uses the
data structure of Okanohara and Sadakane \cite{OS07} for mapping the sorted
$X$ coordinates (where point $X(i)$ is represented as a bit set at position 
$i+X(i)$ in a bitmap of length $n+U$), a similar structure for the $Y$
coordinates, and a wavelet tree for the grid of mapped points. The former 
occupues $n\log\frac{n+U}{n}+O(n) = \log{U+n \choose n} + O(n)$ bits of
space, gives constant-time access to the real coordinate of any point 
$X(i) = select_1(i)-i$, and takes $O(\log\frac{U+n}{n})$ time to map any 
value $x$ to rank space at query time; and 
similarly for $Y$. The wavelet tree requires $n\log n + o(n) = \log n! + O(n)$
bits. Overall, if we ignore the $O(n)$-bit redundancies, the total space is
$\log \left( n!{U+n \choose n}^2 \right)$ bits.



Hence our representation can be in any of $n!{U+n \choose n}^2$ configurations.
Note that we can represent repeated points, which is useful in some cases,
especially when they can have associated values. We show now that our number
of configurations is not much more than the number of possible configurations
of $\Pe$ even if repeated points are forbidden, \ie\ $n$ distinct points from 
$[0,U)^2$ can be in ${U^2 \choose n} \leq {U+n \choose n}^2 n!$ configurations. 
The difference is not so large
because ${U+n \choose n}^2 n!  \leq {U^2 \choose n} c^n$, for any
$c\geq 4$. In terms of bits this means that $2 \log {U+n \choose n} +
\log n! \leq \log {U^2 \choose n} + n \log c$ and therefore our
representation is at most $O(n)$ bits larger than an optimal
representation, aside from the $O(n)$ bits we are already wasting with
respect to  $\log \left( n!{U+n \choose n}^2 \right)$.

To see this, notice that
  ${U+n \choose n}^2 n! = ((U+n)!/(n!U!))^2 n! 
   =  ((U+n)!/U!)^2/n!
   =   (\prod_{i=0}^{n-1}(U+n-i)^2)/n!$.
For sufficiently large $c$, this is 
  $\leq (\prod_{i=0}^{n-1}c(U^2-i))/n! 
   =  c^n U^2!/((U^2-n)!n!)
   =  {U^2 \choose n} c^n$.

We need $c \geq (U+n-i)^2/(U^2-i)$ for any $0 \le i < n$. We next show that, 
if $n \le U$, then $(U+n)^2/U^2 \geq (U + n - i)^2/(U^2-i)$, and thus it
is enough to choose $c \geq (U+n)^2/U^2$. 
%
Simple algebra shows that the condition is equivalent to
$i \le 2 (U+n) - (1+(n/U))^2$. Since we assume for now that
$n/U \leq 1$, the inequality is satisfied if $i \leq 2 (U+n) - 4$. 
Since $i < n$, the inequality always holds (as $U,n \ge 1$).
Thus it is sufficient that $c \ge (U+n)^2/U^2$, which is no larger
than 4 if $n \le U$.

Let us now consider the case $n > U$. Our analysis still holds, up large 
enough $n$. Since $i \leq n-1$, it is sufficient that 
$n-1 \le 2(U+n)-(1+(n/U))^2$. Simple algebra shows this is equivalent
to the cubic inequality $2U^3+nU^2-2nU-n^2 \ge 0$. As a function of $U$
this function has three roots, the only positive one at $U = \sqrt{n}$.
Therefore it is positive for $U \ge \sqrt{n}$, \ie\ $n \le U^2$, which
covers all the possible values of $n$.

\section{Reducing Space for Variance} \label{app:variance}

We now discuss how to bring the $2\lceil\log W\rceil$ space factor
associated to storing weights $w'(p) = w(p)^2$
 closer to $\lceil \log V \rceil$, where $V$ is the overall variance.

Instead of storing $w'(p)$,
 we can store $w''(p) = (w(p) - \lceil T/n \rceil)^2$, where 
 $T/n$ is the average of all the points. To obtain
 $\vara$ from the sum of the $w''$ in $Q$ notice that $(w(p) -
 (T_Q/q))^2 = (w(p) - \lceil T/n \rceil)^2 - 2 (w(p) - \lceil T/n
 \rceil) (\lceil T/n\rceil - (T_Q/q)) + ((T/n) - (T_Q/q))^2$, where
 $T_Q = \soma$ and $q = \conta$. This formula makes use of the stored
 $w''(p)$ values, as well as queries $\soma$ and $\conta$.
Note that rounding is used to keep the values as integers,
 hence limiting the number of bits necessary in its representation.

 To avoid numeric instability and wasted space, it is better that
 $T/n$ is close to $T_q/q$. This simultaneously yields smaller $(w(p)
 - \lceil T/n \rceil)^2$ values and reduces the $(\lceil T/n\rceil -
 (T_Q/q))$ factor in the subtraction. To ensure this
 we may (logically) partition the space and use the local average, instead
 of the global one. Each level of the wavelet tree partitions the
 space into horizontal non-overlapping bands, of the form $[j\cdot
 n/2^k,(j+1)\cdot n/2^k)$, for some $k$. At every level we use the
 average of the band in question. This allows us to compute variances
 for rectangles whose $y$ coordinates are aligned with the bands, while the
 $x$ coordinates are not restricted. For a general rectangle $Q$
 we decompose it into band-aligned rectangles, just as with any other
 query on the wavelet tree (recall Equation~(\ref{eq:wsum})). Alternatively,
  we can use a
 variance update formula~\cite{chan1979updating,R1962} that is stable
 and further reduces the instability of the first calculation. We
 rewrite the update formula of Chan~\ai~\cite{chan1979updating} in
 terms of sets, since originally it was formulated in one dimension.
\begin{lemma}[{\cite[Eq.~2.1]{chan1979updating}}]
  Given two disjoint sets $A$ and $B$, which have values $w(p)$
  associated with element $p$, where $T_A = \somaN(A)$, $T_B =
  \somaN(B)$, $m = \contaN(A)$, $n = \contaN(B)$, $S_A = \sum_{p \in
    A}(w(p) - T_A/m)^2$ and $S_B = \sum_{p \in B} (w(p) - T_B/n)^2$,
  the following equalities hold:
  \begin{eqnarray}
    \label{eq:3}
    T_{A \bigcup B} & = & T_A + T_B \\
    S_{A \bigcup B} & = & S_A + S_B + \frac{m}{n(m+n)}(\frac{n}{m}T_A -
    T_B)^2 \label{eq:5}
  \end{eqnarray}
\end{lemma}
Notice that $\varaN(A) = S_A$. 

\end{document}